\documentclass[12pt,a4paper]{article}

\usepackage[utf8]{inputenc}
\usepackage[T1]{fontenc}
\usepackage[english]{babel}
\usepackage{amsmath,amssymb,amsthm}
\usepackage{mathtools}
\usepackage{dsfont}
\usepackage{mathrsfs}
\usepackage{geometry}
\usepackage{microtype}
\usepackage{setspace}
\usepackage{graphicx}
\usepackage{tikz}
\usetikzlibrary{decorations.pathmorphing}
\usepackage{booktabs}
\usepackage{hyperref}
\hypersetup{
  colorlinks = true,
  linkcolor  = blue,
  citecolor  = blue,
  urlcolor   = blue
}
\usepackage[
  backend    = biber,
  style      = phys,
  sorting    = none,
  maxnames   = 5,
  doi        = true,
  url        = false
]{biblatex}
\usepackage{fancyhdr}
\fancypagestyle{plain}{%
  \fancyhf{}%
  \fancyhead[L]{\scshape A Preprint}%
  \fancyhead[R]{\today}%
  \fancyfoot[C]{\thepage}%
}

\newcommand{\R}{\mathbb{R}}
\newcommand{\diff}{\mathrm{d}}
\newcommand{\bn}{\mathbf{n}}
\newcommand{\br}{\mathbf{r}}
\newcommand{\bR}{\mathbf{R}}

\newcommand{\bE}{\mathbf{E}}
\newcommand{\bB}{\mathbf{B}}
\newcommand{\bH}{\mathbf{H}}
\newcommand{\bD}{\mathbf{D}}
\newcommand{\bP}{\mathbf{P}}
\newcommand{\dOmega}{\partial\Omega}
\newcommand{\Vmoy}{\mathscr{V}(\dOmega)}
\newcommand{\Rp}{R'_\perp}
\newcommand{\Rpp}{R''_\perp}

\newcommand{\chispar}{\chi^s_\parallel}
\newcommand{\chisper}{\chi^s_\perp}
\newcommand{\chisNLpar}{\chi^{(2),s}_\parallel}
\newcommand{\chisNLper}{\chi^{(2),s}_\perp}
\newcommand{\jump}[1]{\left[\! \!\left[#1\right]\! \!\right]}

\newcommand{\Mlin}{\mathbf{M}^{\mathrm{lin}}}
\newcommand{\multiE}[1]{%
  \langle\!\!\langle #1 \rangle\!\!\rangle
}

\newcommand{\brE}{\mathbf{E}}
\newcommand{\brJ}{\mathbf{J}}

\newtheorem{remark}{Remark}
\newtheorem{proposition}{Proposition}
\newtheorem{corollary}{Corollary}

\begin{document}

\title{Where Does Surface \texorpdfstring{$\chi^{(2)}$}{chi(2)} Come From?\\
  A Systematic Derivation of Nonlinear Surface Susceptibilities\\
  from Bulk Nonlocal Response}

\author{F. Zolla}

\date{\today}

\maketitle

\begin{abstract}
We extend the distributional framework developed in the companion
paper~\cite{Zolla2026nonlocalopticalresponsesurface} to the nonlinear case, focusing on the
second-order ($\chi^{(2)}$) response responsible for second-harmonic
generation~(SHG).  Starting from the most general tensorial nonlocal
second-order constitutive relation and combining a spatial moment
expansion with a distributional thin-layer limit, we show that the
full complexity of the nonlinear interfacial response condenses, at
leading order, into two scalars --- the nonlinear surface
susceptibilities $\chisNLpar$ and $\chisNLper$ --- associated with the
tangential and normal components of the electric field, respectively.
A key structural result is established: via a marginal integration
over one field argument, the nonlinear surface problem reduces
recursively to an effective linear one, whose surface susceptibility
is determined by the bulk nonlinear kernel alone.  Generalized nonlinear
Maxwell boundary conditions are derived explicitly for planar and
spherical interfaces, and curvature corrections are obtained
systematically.  The formalism is illustrated on Gaussian, Yukawa, and
tensorial Lorentz kernels.
\end{abstract}

\section{Introduction}

The nonlinear optical response of material interfaces is a subject of
enduring importance~\cite{ZollaWild1,ZollaWild2,ZollaWild3}, from the early theoretical predictions of surface second-harmonic generation~\cite{Bloembergen1962} to its
modern applications in the characterization of metallic
nanostructures~\cite{Bachelier2010,Butet2015} and two-dimensional
materials~\cite{Kumar2013,Li2013}.  At an interface between two
centrosymmetric media, the bulk second-order response vanishes by
symmetry, and the leading nonlinear signal originates entirely from
the surface layer, where inversion symmetry is broken.  This makes
SHG an exquisitely surface-sensitive probe, capable of detecting
monolayer-level structural changes.

The companion paper~\cite{Zolla2026nonlocalopticalresponsesurface} established a systematic
distributional framework for the \emph{linear} nonlocal response of
a material body $\Omega$ with smooth boundary $\dOmega$.  Starting
from the most general tensorial nonlocal constitutive relation for
the polarization $\bP$ and combining a spatial moment expansion with
a distributional thin-layer limit, it was shown that the interfacial
linear response condenses, at leading order, into two scalar surface
susceptibilities $\chispar$ and $\chisper$, which generalize the
Feibelman $d$-parameters to interfaces of arbitrary curvature.

The present paper extends this framework to the nonlinear case.  The
extension is not straightforward: the second-order kernel
$\Delta^{(2)}_{ijk}(\br,\br',\br'')$ depends on \emph{two} spatial
arguments, and the moment expansion must be carried out in a
six-dimensional space.  Nevertheless, the same distributional
machinery applies, and the key structural feature --- the
bulk/boundary decomposition and the thin-layer limit --- carries over
intact.  The central new result is the \emph{marginal integration
principle}: by integrating out one of the two field arguments of
$\Delta^{(2)}_{ijk}$, the nonlinear surface problem reduces to an
effective linear problem with a renormalized kernel, whose surface
susceptibility is determined entirely by the bulk nonlinear kernel.

Throughout this paper, we adopt the same notation as~\cite{Zolla2026nonlocalopticalresponsesurface},
to which we refer freely for definitions and background.  We work at
the degenerate second-harmonic frequency: the two input photons are at
frequency $\omega$, and the output is at $2\omega$.  The
generalization to the non-degenerate case $\omega_1 + \omega_2 \to
\omega_3$ follows by desymmetrization and is indicated in the
conclusion.  The Einstein summation convention is used throughout,
with no distinction between covariant and contravariant indices.

The paper is organized as follows.  Section~\ref{sec:1D_NL} treats
the one-dimensional scalar case, establishing the bulk/boundary
decomposition and the thin-layer limit in the simplest setting.
Section~\ref{sec:3D_NL} extends the analysis to three dimensions,
deriving the nonlinear surface susceptibilities and curvature
corrections.  Section~\ref{sec:marginal} establishes the marginal
integration principle.  Section~\ref{sec:BC_NL} derives the
generalized nonlinear Maxwell boundary conditions.
Section~\ref{sec:examples_NL} illustrates the formalism on explicit
kernel/geometry combinations.  Section~\ref{sec:example_shg} works
out a fully explicit one-dimensional example of surface SHG against
a unit-index background.  Section~\ref{sec:example_yukawa} treats a
non-trivial slab geometry where bulk and surface contributions compete,
exhibiting Maker-fringe oscillations.  Section~\ref{sec:example_centro}
considers the complementary centrosymmetric case, where bulk SHG
vanishes and the response is entirely surface-driven.
Section~\ref{sec:conclu_NL} concludes
and outlines the extension to arbitrary order $\chi^{(n)}$.

\section{The one-dimensional scalar case}
\label{sec:1D_NL}

\subsection{Physical motivation and setup}

As in~\cite{Zolla2026nonlocalopticalresponsesurface}, we begin with the one-dimensional scalar
model, which captures the essential mechanisms without tensorial
complications.  We consider two half-spaces: vacuum for $z < 0$ and
a nonlocal medium occupying $\Omega = \{z > 0\}$, with the interface
at $z = 0$.  The second-order nonlocal constitutive relation reads
\begin{equation}\label{eq:constit1D_NL}
  P^{(2)}(z)
  = \varepsilon_0 \int_{\R}\!\int_{\R}
    \Delta^{(2)}(z,z',z'')\,E(z')\,E(z'')\,
    \diff z'\,\diff z''.
\end{equation}
The kernel $\Delta^{(2)}(z,z',z'')$ is assumed symmetric in its last
two arguments (intrinsic permutation symmetry of the degenerate case):
\begin{equation}\label{eq:symm_1D}
  \Delta^{(2)}(z,z',z'') = \Delta^{(2)}(z,z'',z').
\end{equation}

\begin{remark}
  The kernel $\Delta^{(2)}$ satisfying~\eqref{eq:constit1D_NL} is not
  unique: only its symmetric part in $(z',z'')$ contributes to
  $P^{(2)}$, since the product $E(z')\,E(z'')$ is itself symmetric in
  these two variables.  Among all kernels yielding the same
  polarization, there exists a unique representative symmetric in
  $(z',z'')$---namely, the one satisfying~\eqref{eq:symm_1D}.
  Throughout the sequel, we shall consider only this symmetric
  representative.
\end{remark}

\subsection{Hypotheses on the second-order kernel}

We impose conditions analogous to those of~\cite{Zolla2026nonlocalopticalresponsesurface}.

\begin{enumerate}
\item \textbf{Rapid decay.}  For fixed $z$, the function
  $\Delta^{(2)}(z,z',z'')$ decays rapidly in both $|z'-z|$ and
  $|z''-z|$; all joint moments exist.

\item \textbf{Support in $\Omega \times \Omega \times \Omega$.}
  $\Delta^{(2)}(z,z',z'') = H(z)\,H(z')\,H(z'')\,
  \Delta^{(2)}(z,z',z'')$.

\item \textbf{Bulk homogeneity.}  Far from the interface:
  \begin{equation}\label{eq:homo1D_NL}
    \Delta^{(2)}(z,z',z'')
    = H(z)\,H(z')\,H(z'')\,
      \tilde\Delta^{(2)}(z-z',z-z'').
  \end{equation}

\item \textbf{Centro-symmetry.}  The bulk kernel satisfies
  $\tilde\Delta^{(2)}(-Z',-Z'') = -\tilde\Delta^{(2)}(Z',Z'')$
  (odd under simultaneous sign reversal).
\end{enumerate}

\paragraph{A consequence of~\eqref{eq:symm_1D} and Hypothesis~3.}
The intrinsic permutation symmetry of the full kernel
$\Delta^{(2)}$, equation~\eqref{eq:symm_1D}, combined with bulk
homogeneity, transfers automatically to the bulk kernel:
\begin{equation}\label{eq:bulk_intrinsic_1D}
  \tilde\Delta^{(2)}(Z',Z'') = \tilde\Delta^{(2)}(Z'',Z').
\end{equation}
We shall invoke this property at several places below; it is not an
independent hypothesis.

\begin{remark}
  Hypothesis~4 encodes the fact that in a centrosymmetric bulk medium,
  the second-order bulk response vanishes: applying the inversion
  $\br\to-\br$ to the constitutive relation forces
  $\tilde\Delta^{(2)} = 0$ for a centrosymmetric material.  Here,
  however, we allow a nonzero $\tilde\Delta^{(2)}$ in order to treat
  non-centrosymmetric media; centro-symmetry will be imposed as a
  special case.
\end{remark}

\subsection{Spatial moment expansion}

Setting $Z' = z'-z$ and $Z'' = z''-z$, and expanding $E(z')$ and
$E(z'')$ in Taylor series about $z$:
\begin{equation}\label{eq:Taylor1D_NL}
  E(z+Z') \approx E(z) + Z'\,E'(z) + \tfrac{1}{2}{Z'}^2\,E''(z) + \cdots,
\end{equation}
and inserting into~\eqref{eq:constit1D_NL} yields the
\emph{second-order spatial moment expansion}:
\begin{equation}\label{eq:P_moments1D_NL}
  P^{(2)}(z)
  \approx \varepsilon_0
  \sum_{j,k=0}^{\infty}
  \chi^{(2)}_{jk}(z)\,
  \dfrac{\diff^j E}{\diff z^j}(z)\,
  \dfrac{\diff^k E}{\diff z^k}(z),
\end{equation}
where the \emph{second-order spatial moments} are
\begin{equation}\label{eq:chi2_jk_def}
  \chi^{(2)}_{jk}(z)
  := \frac{1}{j!\,k!}
     \int_{\R}\!\int_{\R}
     {Z'}^j\,{Z''}^k\,
     \tilde\Delta^{(2)}(Z',Z'')\,H(z+Z')\,H(z+Z'')\,
     \diff Z'\,\diff Z''.
\end{equation}

\subsection{Bulk and boundary decomposition}

Following the same procedure as in~\cite{Zolla2026nonlocalopticalresponsesurface}, we split the
integration domain by writing
$\int_{-z}^{+\infty} = \int_{\R} - \int_{-\infty}^{-z}$ for each
variable:
\begin{equation}\label{eq:decomp1D_NL}
  \chi^{(2)}_{jk}(z)
  = \chi^{(2),\Omega}_{jk}\,H(z)
    - \chi^{(2),\Vmoy}_{jk}(z),
\end{equation}
where the \emph{bulk second-order moment} is the constant
\begin{equation}\label{eq:chi2_jk_bulk}
  \chi^{(2),\Omega}_{jk}
  := \frac{1}{j!\,k!}
     \int_{\R}\!\int_{\R}
     {Z'}^j\,{Z''}^k\,
     \tilde\Delta^{(2)}(Z',Z'')\,
     \diff Z'\,\diff Z'',
\end{equation}
and the \emph{second-order boundary term} is
\begin{equation}\label{eq:chi2_jk_bdy}
  \chi^{(2),\Vmoy}_{jk}(z)
  := \frac{H(z)}{j!\,k!}
     \int\!\!\int_{\{Z' < -z\}\cup\{Z'' < -z\}}
     {Z'}^j\,{Z''}^k\,
     \tilde\Delta^{(2)}(Z',Z'')\,
     \diff Z'\,\diff Z''.
\end{equation}

\begin{remark}
  The boundary term $\chi^{(2),\Vmoy}_{jk}(z)$ involves the
  integration domain where at least one of the source points
  $z+Z'$ or $z+Z''$ lies outside $\Omega$ (i.e., in vacuum).
  It is concentrated in a layer of thickness $\sim\ell$ near
  $z=0$ and vanishes as $z\to+\infty$, in exact analogy with
  the linear case~\cite{Zolla2026nonlocalopticalresponsesurface}.
\end{remark}

\subsection{Vanishing of the bulk second-order response in
            centrosymmetric media}

If the medium is centrosymmetric, then $\tilde\Delta^{(2)}(-Z',-Z'')
= -\tilde\Delta^{(2)}(Z',Z'')$, i.e., the kernel is odd under
simultaneous reversal.  In that case, all the bulk moments
$\chi^{(2),\Omega}_{jk}$ with $j+k$ even vanish identically.
In particular, $\chi^{(2),\Omega}_{00} = 0$: there is no bulk
second-harmonic response, as expected from symmetry.

The leading bulk term has $j+k = 1$, i.e., it involves one
field gradient.  However, by the intrinsic
symmetry~\eqref{eq:bulk_intrinsic_1D},
$\chi^{(2),\Omega}_{10} = \chi^{(2),\Omega}_{01}$, so this term
takes the form $\chi^{(2),\Omega}_{10}\,[E(z)\,E'(z) + E'(z)\,E(z)]
= 2\chi^{(2),\Omega}_{10}\,E(z)\,E'(z)$.

\subsection{Distributional thin-layer limit}

We apply the same distributional procedure as in~\cite{Zolla2026nonlocalopticalresponsesurface}.
The leading boundary term at order $j=k=0$ is
\begin{equation}
  \chi^{(2),\Vmoy}_{00}(z)
  = H(z)\int\!\!\int_{\{Z'<-z\}\cup\{Z''<-z\}}
    \tilde\Delta^{(2)}(Z',Z'')\,\diff Z'\,\diff Z''.
\end{equation}
In the thin-layer limit $\ell\to 0$ with $z/\ell$ fixed, this
concentrates near $z=0$ and converges in $\mathcal{D}'$ to
\begin{equation}\label{eq:distrib1D_NL}
  \chi^{(2),\Vmoy}_{00}\!\left(\frac{z}{\eta}\right)
  \;\xrightarrow[\eta\to 0]{\mathcal{D}'}\;
  B_0\,\delta(z),
\end{equation}
with the \emph{leading nonlinear surface coefficient}
\begin{equation}\label{eq:B0_def}
  B_0
  := \int_0^{+\infty}\!\int_0^{+\infty}
     \bigl(Z' + Z''\bigr)\,
     \tilde\Delta^{(2)}(Z',Z'')\,
     \diff Z'\,\diff Z''.
\end{equation}
This is the second-order analogue of the formula
$A_0 = \int_0^{+\infty} Z\,\tilde\Delta(Z)\,\diff Z$
of~\cite{Zolla2026nonlocalopticalresponsesurface}: the nonlinear surface coefficient is set by
the first joint moment of the bulk second-order kernel over the
exterior half-space.

The nonlinear surface polarization at leading order is therefore
\begin{equation}\label{eq:Psurf1D_NL}
  P^{(2),\partial\Omega}(z)
  \approx -\varepsilon_0\,B_0\,
  [E^+(0)]^2\,\delta(z),
\end{equation}
where $E^+(0)$ is the field evaluated at the interface from the
exterior side.

\section{The three-dimensional tensorial case}
\label{sec:3D_NL}

\subsection{General framework and moment expansion}

The most general tensorial second-order nonlocal constitutive
relation reads
\begin{equation}\label{eq:constit3D_NL}
  P^{(2)}_i(\br)
  = \varepsilon_0
    \int_{\R^3}\!\int_{\R^3}
    \Delta^{(2)}_{ijk}(\br,\br',\br'')\,
    E_j(\br')\,E_k(\br'')\,
    \diff\br'\,\diff\br'',
\end{equation}
where the kernel is symmetric in its last two index-position pairs
(intrinsic permutation symmetry):
\begin{equation}\label{eq:symm_3D}
  \Delta^{(2)}_{ijk}(\br,\br',\br'')
  = \Delta^{(2)}_{ikj}(\br,\br'',\br').
\end{equation}
Setting $\bR' := \br'-\br$ and $\bR'' := \br''-\br$, and expanding
the fields about $\br$:
\begin{equation}
  E_j(\br+\bR')
  = E_j(\br)
  + R'_l\,\partial_l E_j(\br)
  + \tfrac{1}{2}R'_l R'_m\,\partial_l\partial_m E_j(\br)
  + \cdots,
\end{equation}
insertion into~\eqref{eq:constit3D_NL} yields the
\emph{second-order spatial moment expansion}:
\begin{multline}\label{eq:P_moments3D_NL}
  P^{(2)}_i(\br)
  = \varepsilon_0 \bigl[
    \chi^{(2,00)}_{ijk}(\br)\,E_j\,E_k \\
    + \chi^{(2,10)}_{ijkl}(\br)\,(\partial_l E_j)\,E_k
    + \chi^{(2,01)}_{ijkl}(\br)\,E_j\,(\partial_l E_k) \\
    + \chi^{(2,11)}_{ijklm}(\br)\,(\partial_l E_j)\,(\partial_m E_k)
    + \cdots \bigr],
\end{multline}
where all fields are evaluated at $\br$, and the moment tensors are
\begin{align}
  \chi^{(2,00)}_{ijk}(\br)
  &:= \int_{\R^3}\!\int_{\R^3}
      \Delta^{(2)}_{ijk}(\br,\br+\bR',\br+\bR'')\,
      \diff\bR'\,\diff\bR'',
  \label{eq:chi200_def}\\[6pt]
  \chi^{(2,10)}_{ijkl}(\br)
  &:= \int_{\R^3}\!\int_{\R^3}
      R'_l\,
      \Delta^{(2)}_{ijk}(\br,\br+\bR',\br+\bR'')\,
      \diff\bR'\,\diff\bR'',
  \label{eq:chi210_def}\\[6pt]
  \chi^{(2,11)}_{ijklm}(\br)
  &:= \int_{\R^3}\!\int_{\R^3}
      R'_l\,R''_m\,
      \Delta^{(2)}_{ijk}(\br,\br+\bR',\br+\bR'')\,
      \diff\bR'\,\diff\bR''.
  \label{eq:chi211_def}
\end{align}

\subsection{Hypotheses on the three-dimensional second-order kernel}

We impose the following conditions in direct analogy
with~\cite{Zolla2026nonlocalopticalresponsesurface}.

\begin{enumerate}

\item \textbf{Rapid decay.}
  For every $\br$, all joint moments of
  $\Delta^{(2)}_{ijk}(\br,\br+\bR',\br+\bR'')$
  with respect to $(\bR',\bR'')$ exist.

\item \textbf{Support in $\Omega^3$.}
  \begin{equation}
    \Delta^{(2)}_{ijk}(\br,\br',\br'')
    = \mathds{1}_\Omega(\br)\,
      \mathds{1}_\Omega(\br')\,
      \mathds{1}_\Omega(\br'')\,
      \Delta^{(2)}_{ijk}(\br,\br',\br'').
  \end{equation}

\item \textbf{Bulk homogeneity.}
  Far from $\dOmega$:
  \begin{equation}\label{eq:homo3D_NL}
    \Delta^{(2)}_{ijk}(\br,\br',\br'')
    = \mathds{1}_\Omega(\br)\,
      \mathds{1}_\Omega(\br')\,
      \mathds{1}_\Omega(\br'')\,
      \tilde\Delta^{(2)}_{ijk}(\br-\br',\br-\br'').
  \end{equation}

\item \textbf{Bulk isotropy.}
  The bulk kernel $\tilde\Delta^{(2)}_{ijk}(\bR',\bR'')$ is isotropic,
  i.e., it is invariant under simultaneous rotation of all three
  arguments.  By the classical structure theorems for isotropic
  tensor functions of vector arguments
  \cite{Weyl1939_ClassicalGroups,Spencer1971_Invariants} (see
  Remark~\ref{rmk:isotropic_basis} below), it admits a finite
  expansion on a basis of \emph{carrier tensors} built from the
  metric $\delta_{ij}$ and the components of $\bR'$ and $\bR''$, with
  scalar coefficients depending only on the joint invariants.
  Three representative low-degree terms are
  \begin{multline}\label{eq:iso_NL}
    \tilde\Delta^{(2)}_{ijk}(\bR',\bR'')
    = f(R',R'',\bR'\cdot\bR'')\,\delta_{ij}\,R''_k \\
    + g(R',R'',\bR'\cdot\bR'')\,\delta_{ik}\,R'_j
    + h(R',R'',\bR'\cdot\bR'')\,\delta_{jk}\,(R'_i + R''_i)
    + \cdots,
  \end{multline}
  where $f$, $g$, $h$ are scalar functions of the invariants
  $R' = |\bR'|$, $R'' = |\bR''|$, and $\bR'\cdot\bR''$, and the
  ``$\cdots$'' stands for the remaining basis elements (cf.\
  Remark~\ref{rmk:isotropic_basis}).

\item \textbf{Centro-symmetry.}
  $\tilde\Delta^{(2)}_{ijk}(-\bR',-\bR'')
  = -\tilde\Delta^{(2)}_{ijk}(\bR',\bR'')$:
  the kernel is odd under simultaneous inversion.  As in the
  one-dimensional case, this implies the vanishing of the bulk
  second-order response for centrosymmetric media.

\end{enumerate}

\begin{remark}[Structure of the isotropic basis]
\label{rmk:isotropic_basis}
  The expansion~\eqref{eq:iso_NL} rests on two classical results
  from invariant theory~\cite{Weyl1939_ClassicalGroups,
  Spencer1971_Invariants}.  First, every polynomial scalar
  invariant of two vectors $(\bR',\bR'')$ under $O(3)$ is a polynomial
  in the three basic invariants
  $R'^2 = \lvert\bR'\rvert^2$, $R''^2 = \lvert\bR''\rvert^2$, and
  $\bR'\cdot\bR''$.  Second, every isotropic tensor-valued function of
  $(\bR',\bR'')$ is a finite linear combination of \emph{carrier
  tensors} formed from the metric $\delta_{ij}$ and the components
  of $\bR'$ and $\bR''$, with coefficients that are scalar functions
  of those invariants.

  For a third-order tensor depending on two vectors, the carrier
  tensors that are linear in each of $\bR'$ and $\bR''$ already form
  the six-element family
  \[
    \delta_{ij}\,R'_k,\quad \delta_{ij}\,R''_k,\quad
    \delta_{ik}\,R'_j,\quad \delta_{ik}\,R''_j,\quad
    \delta_{jk}\,R'_i,\quad \delta_{jk}\,R''_i,
  \]
  to which one must add the cubic terms
  $R'_i R'_j R'_k,\,R'_i R'_j R''_k,\,\ldots,\,R''_i R''_j R''_k$
  (eight in total), and so on for higher polynomial degrees.  The
  display~\eqref{eq:iso_NL} retains only three of these basis
  elements for readability, the omitted ``$\cdots$'' standing for
  the remaining members.  The additional constraints of
  centro-symmetry (Hypothesis~5) and of intrinsic permutation
  symmetry~\eqref{eq:bulk_intrinsic_3D} considerably reduce the
  number of independent scalar coefficients in the full expansion.
  Pseudo-tensorial terms involving the Levi-Civita symbol
  $\epsilon_{ijk}$ are excluded here since
  $\tilde\Delta^{(2)}_{ijk}$ relates polar quantities only
  (cf.\ Appendix~\ref{app:centrosymmetry}).
\end{remark}

\paragraph{A consequence of~\eqref{eq:symm_3D} and Hypothesis~3.}
As in the one-dimensional case, the intrinsic permutation symmetry
of the full kernel, equation~\eqref{eq:symm_3D}, combined with bulk
homogeneity, implies that the bulk kernel inherits the corresponding
symmetry:
\begin{equation}\label{eq:bulk_intrinsic_3D}
  \tilde\Delta^{(2)}_{ijk}(\bR',\bR'')
  = \tilde\Delta^{(2)}_{ikj}(\bR'',\bR').
\end{equation}
This is therefore not an independent hypothesis.

\begin{remark}
  The last two assumptions can be relaxed if one wishes to treat
  non-centro\-symmetric or anisotropic media.  In that case, additional
  terms appear in the moment expansion, but the distributional
  framework developed below applies without structural change.
\end{remark}

\subsection{Bulk and boundary decomposition}

For $\br\in\Omega$, the integration domain decomposes as
$\R^6 = (\Omega-\br)^2 \cup \text{(complement)}$.  Following the
same splitting as in~\cite{Zolla2026nonlocalopticalresponsesurface}, the second-order polarization
splits into a bulk part and an interfacial correction:
\begin{equation}\label{eq:decomp_P3D_NL}
  P^{(2)}_i(\br)
  = P^{(2),\Omega}_i(\br)
  + P^{(2),\Vmoy}_i(\br),
\end{equation}
where the \emph{bulk second-order polarization} is
\begin{equation}\label{eq:Pbulk3D_NL}
  P^{(2),\Omega}_i(\br)
  := \varepsilon_0\,\mathds{1}_\Omega(\br)
     \int_{\R^3}\!\int_{\R^3}
     \tilde\Delta^{(2)}_{ijk}(\bR',\bR'')\,
     E_j(\br+\bR')\,E_k(\br+\bR'')\,
     \diff\bR'\,\diff\bR''
\end{equation}
(that of an infinite homogeneous nonlinear medium), and the
\emph{second-order interfacial contribution} is
\begin{equation}\label{eq:deltaP3D_NL}
  P^{(2),\Vmoy}_i(\br)
  := -\varepsilon_0\,\mathds{1}_\Omega(\br)
     \int\!\!\int_{\Omega^c\times\Omega
                   \cup\,\Omega\times\Omega^c}
     \tilde\Delta^{(2)}_{ijk}(\br-\br',\br-\br'')\,
     E^+_j(\br')\,E_k(\br'')\,
     \diff\br'\,\diff\br'',
\end{equation}
concentrated in a layer of thickness $\ell$ near $\dOmega$.

\paragraph{Hierarchy of the boundary cells.}
The integration domain $\R^6$ partitions in fact into \emph{four}
cells, $\Omega\times\Omega$, $\Omega^c\times\Omega$,
$\Omega\times\Omega^c$, and $\Omega^c\times\Omega^c$
(Fig.~\ref{fig:decomp_termes_croises}).  The first generates the
bulk polarization~\eqref{eq:Pbulk3D_NL}; the two ``cross cells'',
where exactly one source point lies outside $\Omega$, are the
contributions retained in~\eqref{eq:deltaP3D_NL}; the remaining
``both-outside'' cell $\Omega^c\times\Omega^c$ has been omitted.

This omission is justified by a hierarchy in the small parameter
$\ell$.  By the rapid-decay hypothesis (Hypothesis~1), each
constraint forcing a source point to lie outside $\Omega$ costs a
factor of order $e^{-d/\ell}$, where $d$ denotes the distance from
$\br$ to $\dOmega$.  After the distributional thin-layer limit, each
cross cell therefore yields a $\delta_{\dOmega}$ contribution at
order $\ell^1$ : this is the leading surface susceptibility
derived in~\S\ref{sec:distrib3D_NL} below.  The both-outside cell,
carrying \emph{two} such exterior constraints, contributes only at
order $\ell^2$, on the same footing as the genuine curvature
corrections, into which it is absorbed in the $O(\ell^2)$ terms
of~\eqref{eq:deltaP_NL_distrib}.

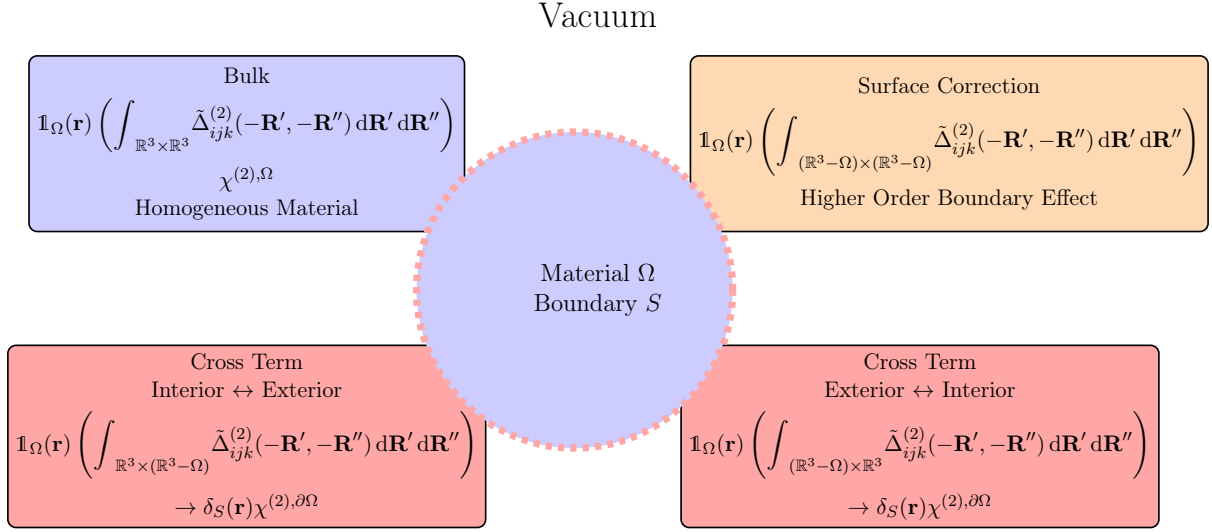
\begin{figure}[h]
  \begin{center}
    \resizebox{\textwidth}{!}{\begin{tikzpicture}[
box/.style={rounded corners, draw, thick, align=center, minimum width=6cm, minimum height=3cm, font=\small},
crossbox/.style={rounded corners, draw, thick, align=center, minimum width=6cm, minimum height=3cm, font=\small},
title/.style={font=\bfseries\large},
formula/.style={font=\large}
]

\node[box, fill=blue!20] (bulk) at (-7,2.5)
{
{\title {}Bulk}\\[6pt]
$
\mathds{1}_{\Omega}(\mathbf{r})
\left(
\displaystyle{\int}_{\R^3 \times \R^3} \tilde{\Delta}_{ijk}^{(2)}(-\mathbf{R}',-\mathbf{R}'') \, \diff \mathbf{R}' \, \diff \mathbf{R}''
\right)
$\\[6pt]
$\chi^{(2),\Omega}$\\
Homogeneous Material
};

\node[box, fill=orange!30] (surf) at (5,2.5)
{
{\title {}Surface Correction}\\[6pt]
$\mathds{1}_{\Omega}(\mathbf{r})
\left(
\displaystyle{\int}_{(\R^3-\Omega) \times (\R^3- \Omega)} \tilde{\Delta}_{ijk}^{(2)}(-\mathbf{R}',-\mathbf{R}'') \, \diff \mathbf{R}' \, \diff \mathbf{R}''
\right)
$\\[6pt]
Higher Order Boundary Effect
};

\node[crossbox, fill=red!35] (cross1) at (-7,-2.5)
{
{\title {}Cross Term}\\
Interior $\leftrightarrow$ Exterior\\[6pt]
$\mathds{1}_{\Omega}(\mathbf{r})
\left(
\displaystyle{\int}_{\R^3\times (\R^3-\Omega)} \tilde{\Delta}_{ijk}^{(2)}(-\mathbf{R}',-\mathbf{R}'') \, \diff \mathbf{R}' \, \diff \mathbf{R}''
\right)
$\\[6pt]
$\rightarrow \delta_S(\mathbf r)\chi^{(2),\partial \Omega}$
};

\node[crossbox, fill=red!35] (cross2) at (4.5,-2.5)
{
{\title {}Cross Term}\\
Exterior $\leftrightarrow$ Interior\\[6pt]
$\mathds{1}_{\Omega}(\mathbf{r})
\left(
\displaystyle{\int}_{(\R^3-\Omega) \times \R^3} \tilde{\Delta}_{ijk}^{(2)}(-\mathbf{R}',-\mathbf{R}'') \, \diff \mathbf{R}' \, \diff \mathbf{R}''
\right)
$\\[6pt]
$\rightarrow \delta_S(\mathbf r)\chi^{(2),\partial \Omega}$
};

\def\R{2.7} 
\def\ep{3.5pt} 

\filldraw[
    fill=blue!20,
    draw=red!35,
    dashed,
    line width=\ep
] (-1.4,0) circle (\R);

\node at (-1.0,0)
{
\begin{tabular}{c}
\textcolor{black}{Materia}l $\Omega$\\
Boundary $S$
\end{tabular}
};

\node at (-1,4.7) {\Large Vacuum};

\end{tikzpicture}}
  \end{center}
  \caption{Schematic decomposition of the second-order polarization
    at a point $\br$ near the boundary $\dOmega$.  The integration
    domain over the two source variables $(\br',\br'')$ splits into
    four cells: the homogeneous bulk contribution (top-left), a
    higher-order surface correction where both source points lie
    outside~$\Omega$ (top-right), and two cross terms (bottom)
    where exactly one of the source points lies outside.  After
    the distributional thin-layer limit, the two cross terms
    combine to produce the leading surface susceptibility
    $\chi^{(2),\partial\Omega}$ supported on $\dOmega$.}
  \label{fig:decomp_termes_croises}
\end{figure}

Accordingly, each moment tensor decomposes as
\begin{equation}\label{eq:decomp_chi3D_NL}
  \chi^{(2,pq)}_{ij\cdots}(\br)
  = \chi^{(2,pq),\Omega}_{ij\cdots}
  - \chi^{(2,pq),\Vmoy}_{ij\cdots}(\br),
\end{equation}
with constant bulk moments and boundary terms concentrated near
$\dOmega$.

\paragraph{Vanishing of the bulk response for centrosymmetric media.}
By centro-symmetry, $\tilde\Delta^{(2)}_{ijk}(-\bR',-\bR'') =
-\tilde\Delta^{(2)}_{ijk}(\bR',\bR'')$.  The bulk moment
$\chi^{(2,00),\Omega}_{ijk} = \int\!\int
\tilde\Delta^{(2)}_{ijk}\,\diff\bR'\,\diff\bR''$ changes sign under
$(\bR',\bR'')\to(-\bR',-\bR'')$ and is therefore zero: the bulk
second-harmonic response vanishes identically for centrosymmetric
media.

\subsection{Distributional thin-layer limit and nonlinear surface
            susceptibilities}
\label{sec:distrib3D_NL}

We apply the distributional thin-layer procedure
of~\cite{Zolla2026nonlocalopticalresponsesurface} to the boundary terms~\eqref{eq:decomp_chi3D_NL}.
The leading-order moment ($p=q=0$) yields a distribution supported
on $\dOmega$:
\begin{equation}\label{eq:M0_NL}
  \chi^{(2,00),\Vmoy}_{ijk}(\br)
  \;\xrightarrow{\ell\to 0}\;
  \mathcal{M}^{(2,00)}_{ijk}\,\delta_{\dOmega},
\end{equation}
with the \emph{nonlinear surface moment tensor}
\begin{equation}\label{eq:M0_NL_def}
  \mathcal{M}^{(2,00)}_{ijk}(\bR_0)
  :=
  \int_0^{+\infty}\!\int_{\R^2}
  \int_0^{+\infty}\!\int_{\R^2}
  \tilde\Delta^{(2)}_{ijk}(\bR',\bR'')\,
  \diff^2\bR'_\parallel\,\diff\Rp\,
  \diff^2\bR''_\parallel\,\diff\Rpp,
\end{equation}
where both half-spaces $\Rp > 0$ and $\Rpp > 0$ correspond to the
exterior $\Omega^c$ near $\dOmega$.

\subsubsection{Decomposition onto interface projectors}

By bulk isotropy, $\mathcal{M}^{(2,00)}_{ijk}$ decomposes onto the
natural projectors of the interface.  For a centrosymmetric medium,
the leading nonzero moment arises from the $p+q = 1$ terms (one
gradient), and the tensor $\mathcal{S}^{(2,0)}_{ijk}$ decomposes as:
\begin{equation}\label{eq:S20_decomp}
  \mathcal{S}^{(2,0)}_{ijk}
  = \chisNLpar\,P^\parallel_{ij}\,n_k
  + \chisNLper\,n_i\,n_j\,n_k
  + \chisNLpar[2]\,\delta_{ij}\,n_k,
\end{equation}
where $P^\parallel_{ij} = \delta_{ij} - n_i n_j$ is the tangential
projector, $\bn$ is the outward unit normal, and $\chisNLpar$,
$\chisNLper$, $\chisNLpar[2]$ are three independent scalar surface
susceptibilities given by integrals of
$\tilde\Delta^{(2)}_{ijk}$ over the exterior half-space.

\begin{remark}[Comparison with the linear case]
  In the linear case~\cite{Zolla2026nonlocalopticalresponsesurface}, isotropy and centro-symmetry
  together force $\mathcal{S}^{(1)}_{ijk} = 0$ at order $\ell^1$.
  In the nonlinear case, centro-symmetry forces the bulk
  $\chi^{(2,00),\Omega}_{ijk} = 0$, but the \emph{surface} term
  $\mathcal{S}^{(2,0)}_{ijk}$ is generically nonzero: the symmetry
  breaking at the interface generates a nonlinear surface response
  even in a centrosymmetric bulk medium.  This is the fundamental
  origin of surface SHG.
\end{remark}

\subsubsection{Curvature corrections}

At order $\ell^2$, the mean curvature $H$ and Gaussian curvature $K$
generate corrections to the nonlinear surface susceptibilities,
exactly as in the linear case:
\begin{equation}\label{eq:deltaP_NL_distrib}
  P^{(2),\dOmega}_i(\br)
  = -\varepsilon_0
  \Bigl[
    \mathcal{S}^{(2,0)}_{ijk}\,\delta_{\dOmega}
    + \ell^2\,\bigl(
      \nu^{(2),\parallel} H\,P^\parallel_{ij}\,n_k
      + \nu^{(2),\perp} H\,n_i n_j n_k
      \bigr)\,\delta_{\dOmega}
    + \cdots
  \Bigr]
  E^+_j\,E^+_k
  + O(\ell^3),
\end{equation}
where the curvature coefficients $\nu^{(2),\parallel}$ and
$\nu^{(2),\perp}$ are given by the second-order joint moments of
$\tilde\Delta^{(2)}_{ijk}$ over the exterior half-space.

\section{The marginal integration principle}
\label{sec:marginal}

We now establish the key structural result announced in the
introduction: the nonlinear surface problem reduces, via a marginal
integration, to an effective linear one.

\subsection{Statement of the principle}

\begin{proposition}[Marginal integration]\label{prop:marginal}
  Define the \emph{marginalized kernel}
  \begin{equation}\label{eq:marginal_def}
    \tilde\Delta^{(2\to 1)}_{ij}(\bR';\,E^+)
    := \int_{\R^3}
       \tilde\Delta^{(2)}_{ijk}(\bR',\bR'')\,E^+_k(\br+\bR'')\,
       \diff\bR''.
  \end{equation}
  Then the second-order bulk polarization~\eqref{eq:Pbulk3D_NL}
  can be written as
  \begin{equation}\label{eq:Pmarginal}
    P^{(2),\Omega}_i(\br)
    = \varepsilon_0\,\mathds{1}_\Omega(\br)
      \int_{\R^3}
      \tilde\Delta^{(2\to 1)}_{ij}(\bR';\,E^+)\,
      E_j(\br+\bR')\,
      \diff\bR',
  \end{equation}
  which has the form of a \emph{linear} nonlocal constitutive
  relation with the field-dependent effective kernel
  $\tilde\Delta^{(2\to 1)}_{ij}$.
\end{proposition}

\begin{proof}
  Direct substitution of~\eqref{eq:marginal_def}
  into~\eqref{eq:Pmarginal} and use of
  Fubini's theorem.
\end{proof}

\subsection{Consequences for the surface susceptibility}

Applying the distributional thin-layer procedure to the marginalized
kernel, the nonlinear surface susceptibility $\mathcal{S}^{(2,0)}_{ijk}$
can be expressed as
\begin{equation}\label{eq:S20_marginal}
  \mathcal{S}^{(2,0)}_{ijk}\,E^+_j\,E^+_k
  = \mathcal{S}^{(1,\mathrm{eff})}_{ij}[E^+]\,E^+_j,
\end{equation}
where
\begin{equation}\label{eq:S1eff}
  \mathcal{S}^{(1,\mathrm{eff})}_{ij}[E^+]
  := \int_0^{+\infty}\!\int_{\R^2}
     \tilde\Delta^{(2\to 1)}_{ij}(\bR';\,E^+)\,
     \diff^2\bR'_\parallel\,\diff\Rp
\end{equation}
is the surface susceptibility of the effective linear problem with
kernel $\tilde\Delta^{(2\to 1)}_{ij}$.

\begin{corollary}
  All results of the linear theory~\cite{Zolla2026nonlocalopticalresponsesurface} --- the surface
  susceptibility formulas, the curvature corrections, the generalized
  Maxwell boundary conditions --- apply to the nonlinear problem
  through the substitution
  $\tilde\Delta^{(1)}_{ij} \to \tilde\Delta^{(2\to 1)}_{ij}(\cdot;E^+)$.
  The nonlinear surface problem is thus closed recursively: it reduces
  to a $\chi^{(1)}$ effective problem with a field-dependent kernel.
\end{corollary}

\begin{remark}[Extension to order $n$]
  The same argument applies to the $n$-th order response.  The
  $n$-th order bulk kernel $\tilde\Delta^{(n)}_{ij_1\cdots j_n}$
  is marginalized over $n-1$ field arguments to give an effective
  linear kernel $\tilde\Delta^{(n\to 1)}_{ij}$, and the surface
  susceptibility of the $n$-th order problem is expressed in terms
  of that of a $\chi^{(1)}$ effective problem.  The hierarchy is
  closed at every order.
\end{remark}

\section{Generalized nonlinear Maxwell boundary conditions}
\label{sec:BC_NL}

\subsection{Structure of the second-harmonic polarization}

At frequency $2\omega$, the material response is driven by the
fundamental field at $\omega$.  We work in the undepleted-pump
approximation: the fundamental field $\bE^\omega$ satisfies the
linear Maxwell equations with the generalized boundary conditions
of~\cite{Zolla2026nonlocalopticalresponsesurface}, and the second-harmonic field $\bE^{2\omega}$
is driven by the nonlinear polarization $\bP^{(2)}$ as a source.

The total second-harmonic polarization splits as
in~\eqref{eq:decomp_P3D_NL}:
\begin{equation}\label{eq:P2omega_total}
  \bP^{(2)}(\br)
  = \bP^{(2),\Omega}(\br)
  + \bP^{(2),\dOmega}(\br),
\end{equation}
where $\bP^{(2),\dOmega}$ is the interfacial contribution given
by~\eqref{eq:deltaP_NL_distrib}.

\subsection{Maxwell equations at $2\omega$}

The macroscopic Maxwell equations at $2\omega$ read:
\begin{align}
  \nabla\times\bE^{2\omega}
  &= 2i\omega\mu_0\bH^{2\omega},
  \label{eq:MaxFar_2w}\\
  \nabla\times\bH^{2\omega}
  &= -2i\omega\bigl(\bD^{(1),2\omega} + \bP^{(2)}\bigr),
  \label{eq:MaxAmp_2w}\\
  \nabla\cdot\bigl(\bD^{(1),2\omega} + \bP^{(2)}\bigr)
  &= 0,
  \label{eq:MaxGauss_2w}
\end{align}
where $\bD^{(1),2\omega} = \varepsilon_0\varepsilon(2\omega)\bE^{2\omega}$
is the linear displacement at $2\omega$.  The singular distribution
$\bP^{(2),\dOmega}$ supported on $\dOmega$ drives the jump
conditions for the second-harmonic field.

\subsection{Identification of boundary conditions}

Inserting~\eqref{eq:deltaP_NL_distrib} into the Maxwell
equations~\eqref{eq:MaxFar_2w}--\eqref{eq:MaxGauss_2w} and
applying the distributional differentiation rules
of~\cite{Zolla2026nonlocalopticalresponsesurface}, we identify the coefficients of
$\delta_{\dOmega}$ in each equation.

\paragraph{Faraday's law.}
The tangential second-harmonic electric field is continuous:
\begin{equation}\label{eq:BC_Far_2w}
  \bn\times\jump{\bE^{2\omega}} = \mathbf{0}.
\end{equation}

\paragraph{Amp\`{e}re's law.}
Collecting the $\delta_{\dOmega}$ coefficient:
\begin{equation}\label{eq:BC_Amp_2w}
  \bn\times\jump{\bH^{2\omega}}
  = -2i\omega\varepsilon_0\,
    \mathcal{S}^{(2,0)}_{ijk}\,n_k\,
    E^{+,\omega}_j\,E^{+,\omega}_k
  + O(\ell^2),
\end{equation}
where $E^{+,\omega}_j$ is the fundamental field evaluated at the
interface from the exterior side.

\paragraph{Gauss's law.}
Collecting the $\delta_{\dOmega}$ coefficient:
\begin{equation}\label{eq:BC_Gauss_2w}
  \jump{D^{(1),2\omega}_\perp}
  = -\varepsilon_0
    \Bigl[
      \bigl(\chisNLper + \nu^{(2),\perp} H\bigr)
      E^{+,\omega}_\perp\,E^{+,\omega}_\perp
      + \nabla_s\cdot
        \bigl(\chisNLpar\,\bE^{+,\omega}_\parallel
              \,E^{+,\omega}_\perp\bigr)
    \Bigr]
  + O(\ell^2).
\end{equation}

\begin{remark}[Structure of the boundary conditions]
  The nonlinear boundary conditions~\eqref{eq:BC_Amp_2w}
  and~\eqref{eq:BC_Gauss_2w} have the same structure as their
  linear counterparts~\cite{Zolla2026nonlocalopticalresponsesurface}: a jump in the tangential
  magnetic field and in the normal displacement, both driven by a
  surface source term.  The key difference is that the source is
  quadratic in the fundamental field $\bE^{+,\omega}$ rather than
  linear in the field at $2\omega$.
\end{remark}

\subsection{Special case: planar interface}

For $\dOmega = \{z=0\}$ with $H = K = 0$, the boundary conditions
reduce to:
\begin{align}
  \bn\times\jump{\bH^{2\omega}}
  &= -2i\omega\varepsilon_0\,
     \chisNLpar\,
     \bigl(E^{+,\omega}_\parallel\bigr)^2,
  \label{eq:BC_plane_H_2w}\\[6pt]
  \jump{D^{(1),2\omega}_\perp}
  &= -\varepsilon_0\,\chisNLper\,
     \bigl(E^{+,\omega}_\perp\bigr)^2
  - \varepsilon_0\,\chisNLpar\,
     \nabla_s\cdot
     \bigl(\bE^{+,\omega}_\parallel
           E^{+,\omega}_\perp\bigr).
  \label{eq:BC_plane_D_2w}
\end{align}
For $\chisNLpar = \chisNLper = 0$ (no surface),
conditions~\eqref{eq:BC_plane_H_2w}--\eqref{eq:BC_plane_D_2w}
reduce to the standard continuity conditions at $2\omega$,
consistently with the absence of bulk SHG in a centrosymmetric
medium.

\subsection{Special case: spherical interface}

For the sphere of radius $R_0$ ($H = 2/R_0$,
$\bn = \hat{\mathbf{r}}$):
\begin{align}
  \bn\times\jump{\bH^{2\omega}}
  &= -2i\omega\varepsilon_0\,
     \left(\chisNLpar
           + \frac{2\nu^{(2),\parallel}}{R_0}\right)
     \bigl(\bE^{+,\omega}_\parallel\bigr)^2
  + O\!\left(\frac{\ell^2}{R_0^2}\right),
  \label{eq:BC_sphere_H_2w}\\[6pt]
  \jump{D^{(1),2\omega}_\perp}
  &= -\varepsilon_0
     \left(\chisNLper
           + \frac{2\nu^{(2),\perp}}{R_0}\right)
     \bigl(E^{+,\omega}_\perp\bigr)^2
  + O\!\left(\frac{\ell^2}{R_0^2}\right).
  \label{eq:BC_sphere_D_2w}
\end{align}
The terms proportional to $1/R_0$ are the \emph{nonlinear Mie
corrections}: they shift the SHG intensity pattern of a plasmonic
nanosphere with respect to the flat-interface result, and become
significant for $R_0 \lesssim 10\ell$.

\section{Explicit kernel/surface combinations}
\label{sec:examples_NL}

We illustrate the general formulas on three analytically tractable
kernels, for the planar interface.

\subsection{Gaussian scalar kernel}

\paragraph{Kernel.}
$\tilde\Delta^{(2)}_{ijk}(\bR',\bR'')
= \dfrac{A_0^{(2)}}{\ell^5}\,e^{-(R'^2+R''^2)/\ell^2}\,
  (\delta_{ij}\,R''_k + \delta_{ik}\,R'_j)$,
where $A_0^{(2)} := \Delta_0^{(2)}\ell^5$ is the bulk amplitude
($\ell$-independent),
which satisfies centro-symmetry and intrinsic permutation
symmetry by construction.

\paragraph{Nonlinear surface susceptibilities.}
Performing the Gaussian integrals over the exterior half-space
$\Rp > 0$, $\Rpp > 0$:
\begin{equation}
  \chisNLpar
  = \chisNLper
  = \frac{\pi\,A_0^{(2)}}{4}.
\end{equation}
The equality $\chisNLpar = \chisNLper$ reflects the isotropy of
the scalar kernel, in direct analogy with the linear
case~\cite{Zolla2026nonlocalopticalresponsesurface}.

\paragraph{Second-harmonic surface polarization (planar).}
\begin{equation}
  P^{(2),\dOmega}_i
  = -\frac{\varepsilon_0\pi\,A_0^{(2)}}{4}\,
    \delta_{ij}\,\bigl(E^{+,\omega}_j\bigr)^2\,\delta(z).
\end{equation}

\subsection{Yukawa scalar kernel on a spherical interface}

\paragraph{Kernel.}
$\tilde\Delta^{(2)}_{ijk}(\bR',\bR'')
= \dfrac{A_0^{(2)}}{\ell^4}\,e^{-(R'+R'')/\ell}\,
  (R'\,R'')^{-1}\,(\delta_{ij}\,R''_k + \delta_{ik}\,R'_j)$,
with the same convention $A_0^{(2)} := \Delta_0^{(2)}\ell^5$.

\paragraph{Nonlinear surface susceptibilities.}
\begin{equation}
  \chisNLpar = \chisNLper = 2\pi\,A_0^{(2)}.
\end{equation}

\paragraph{Curvature correction on the sphere.}
From~\eqref{eq:BC_sphere_H_2w} with $H = 2/R_0$:
\begin{equation}
  \frac{|\delta P^{(2),\mathrm{curv}}|}
       {|\delta P^{(2),(0)}|}
  = \frac{2\nu^{(2),\parallel}}{R_0\,\chisNLpar}
  \sim \frac{12\ell^2}{R_0}.
\end{equation}
For $R_0 = 10$\,nm and $\ell = 0.3$\,nm, this gives
$\approx 1.1\times 10^{-2}$: the nonlinear Mie correction is of
order~$1\%$.

\subsection{Tensorial Lorentz kernel}

\paragraph{Kernel.}
\begin{equation}
  \tilde\Delta^{(2)}_{ijk}(\bR',\bR'')
  = \frac{A_0^{(2)}}{\ell^5\,(1+R'^2/\ell^2)^2(1+R''^2/\ell^2)^2}
    \left[
      \delta_{ij}\,R''_k
      + \delta_{ik}\,R'_j
      + \frac{\beta}{R'^2+\ell^2}
        \left(R'_i R'_j R''_k + R'_i R''_j R'_k\right)
    \right].
\end{equation}

\paragraph{Nonlinear surface susceptibilities (plane).}
\begin{align}
  \chisNLpar
  &= \frac{2\pi^4\,A_0^{(2)}(3+\beta)}{9},\\[4pt]
  \chisNLper
  &= \frac{\pi^4\,A_0^{(2)}(6+\beta)}{9}.
\end{align}
The surface anisotropy
$\chisNLper - \chisNLpar = -\pi^4A_0^{(2)}\beta/9$
has the same origin as in the linear case: it reflects the
tensorial structure of the bulk kernel, not the interface
geometry.

\begin{table}[ht]
\centering
\renewcommand{\arraystretch}{1.8}
\begin{tabular}{llcc}
\toprule
Kernel & Surface & $\chisNLpar$ & $\chisNLper$ \\
\midrule
Gaussian scalar
  & Plane
  & $\dfrac{\pi\,A_0^{(2)}}{4}$
  & $\dfrac{\pi\,A_0^{(2)}}{4}$ \\[8pt]
Yukawa scalar
  & Sphere $R_0$
  & $2\pi\,A_0^{(2)}$
  & $2\pi\,A_0^{(2)}$ \\[8pt]
Lorentz tensorial ($\beta$)
  & Plane
  & $\dfrac{2\pi^4\,A_0^{(2)}(3+\beta)}{9}$
  & $\dfrac{\pi^4\,A_0^{(2)}(6+\beta)}{9}$ \\
\bottomrule
\end{tabular}
\caption{Nonlinear surface susceptibilities $\chisNLpar$ and $\chisNLper$
  for the three kernel/surface combinations treated explicitly. Scalar
  kernels (Gaussian, Yukawa) yield isotropic surface response
  ($\chisNLpar=\chisNLper$); the tensorial Lorentz kernel breaks this
  isotropy through the anisotropy parameter~$\beta$, which controls the
  surface anisotropy
  $\chisNLper-\chisNLpar=-\pi^4A_0^{(2)}\beta/9$.}
\label{tab:NL_explicit}
\end{table}

\section{A simple worked example: surface SHG against a
         unit-index background}
\label{sec:example_shg}

To illustrate the formalism in its simplest setting, we briefly
revisit the case treated
in~\parencite{ZollaWild1,ZollaWild2,ZollaWild3}, in which the
entire second-harmonic signal is generated by a singular surface
$\chi^{(2)}$ embedded in a background of unit refractive index.
The starting point is the coupled system at the fundamental
frequency $\omega$ and at the second-harmonic frequency
$2\omega$:
\begin{equation}\label{sys:chi2:d1=2}
  \left\{
  \begin{array}{l}
    \Mlin_\omega\,\brE^\omega = -i\,\omega\mu_0\,\brJ^\omega, \\[6pt]
    \Mlin_{2\omega}\,\brE^{2\omega} + \dfrac{(2\omega)^2}{c^2}\,
    \multiE{\brE^\omega,\brE^\omega} = 0,
  \end{array}
  \right.
\end{equation}
where $\Mlin_\omega$ and $\Mlin_{2\omega}$ are the linear Maxwell operators at
$\omega$ and $2\omega$ respectively, and
$\multiE{\cdot\,,\cdot}$ denotes the bilinear nonlinear coupling
driven by the surface $\chi^{(2)}$. To be more precise.
\[
 \Mlin_{p\omega} \brE^{p\omega} := - \nabla \times \nabla \times \brE^{p\omega} + \frac{(p \, \omega)^2}{c^2} \, \varepsilon_r(\mathbf{s},p \, \omega) \, \brE^{p\omega} \, .
 \]
 and where $\multiE{\brE^{\omega}, \brE^{\omega}} $ is the following bilinear form~:
 \begin{equation}
\multiE{\brE^{\omega}, \brE^{\omega}}= \chi_{(2)}(\omega,\omega)  {:} \brE^{\omega}\otimes \brE^{\omega} \, ,
\end{equation}

\paragraph{One-dimensional reduction.}
Restricting to one dimension with normal incidence and an
\emph{ad hoc} polarization, and assuming for simplicity that the
linear refractive index equals~$1$ everywhere,
system~\eqref{sys:chi2:d1=2} reduces to
\begin{equation}\label{sys:chi2:d1=2:scalar}
  \left\{
  \begin{array}{l}
    u_1''(x) + k_0^2\,u_1(x) = 0, \\[6pt]
    u_2''(x) + (2k_0)^2\,u_2(x)
    + (2k_0)^2\,u_1^2(0)\,\chi_0^{(2),S}\,\delta(x) = 0,
  \end{array}
  \right.
\end{equation}
the surface $\chi^{(2)}$ being localised at $x=0$.

\paragraph{Solution.}
Imposing outgoing-wave conditions, the first equation
in~\eqref{sys:chi2:d1=2:scalar} gives $u_1(x) = A\,e^{ik_0 x}$, and
the second equation becomes
\[
  u_2''(x) + 4k_0^2\,u_2(x)
  + 4k_0^2\,A^2\,\chi_0^{(2),S}\,\delta(x) = 0.
\]
The solution takes two distinct forms according to the sign of
$x$:
\[
  u_2^-(x) = A_2^-\,e^{-2ik_0 x} \quad (x<0),
  \qquad
  u_2^+(x) = A_2^+\,e^{+2ik_0 x} \quad (x>0).
\]
The matching conditions at $x=0$, namely $\jump{u_2}_0 = 0$ and
$\jump{u_2'}_0 = -4k_0^2\,A^2\,\chi_0^{(2),S}$, yield
$A_2^+ = A_2^-$ and finally
\[
  A_2^+ \;=\; A_2^- \;=\; i\,k_0\,A^2\,\chi_0^{(2),S}.
\]
In summary,
\begin{equation}\label{eq:u2_final_example}
  u_2(x) \;=\; i\,k_0\,A^2\,\chi_0^{(2),S}\,e^{2ik_0|x|}.
\end{equation}
The result is consistent with the outgoing-wave Green function
$G(x)=e^{2ik_0|x|}/(4ik_0)$ for the operator
$\partial_x^2 + (2k_0)^2$, since
$u_2(x) = -4k_0^2\,A^2\,\chi_0^{(2),S}\,G(x)
       = i\,k_0\,A^2\,\chi_0^{(2),S}\,e^{2ik_0|x|}$.

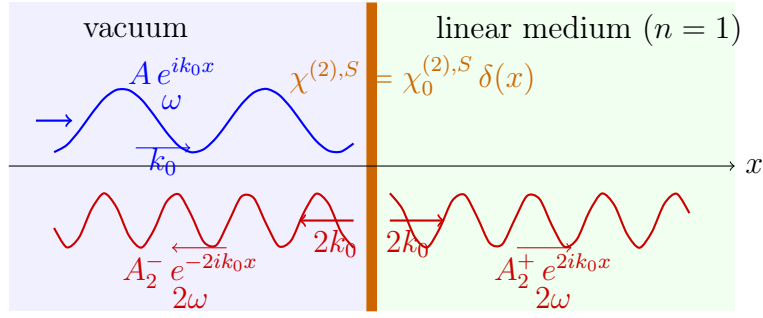
\begin{figure}[h]
  \begin{center}
    \begin{tikzpicture}[scale=1.2]


\fill[blue!6] (-4,-1.6) rectangle (0,1.8);
\fill[green!6] (0.,-1.6) rectangle (4,1.8);

\draw[line width=4pt,orange!80!black] (0,-1.6) -- (0,1.8);

\draw[->] (-4,0) -- (4,0) node[right] {$x$};

\node at (-2.6,1.5) {vacuum};
\node at (2.4,1.5) {linear medium ($n=1$)};

\node[orange!80!black] at (0.45,1.0) {$\chi^{(2),S}=\chi_0^{(2),S} \, \delta(x)$};


\draw[domain=-3.5:-0.2,smooth,variable=\x,blue,thick]
plot ({\x},{0.35*sin(4*\x r)+0.5});

\draw[->,blue,thick] (-3.7,0.5) -- (-3.3,0.5);

\node[blue] at (-2.2,1.05) {$A\, e^{ik_0 x}$};
\node[blue] at (-2.2,0.7) {$\omega$};

\draw[->,blue] (-2.6,0.2) -- (-2.0,0.2);
\node[blue] at (-2.3,0.05) {$k_0$};


\draw[domain=-0.2:-3.5,smooth,variable=\x,red!80!black,thick]
plot ({\x},{0.30*sin(8*\x r)-0.6});

\draw[->,red!80!black,thick] (-0.2,-0.6) -- (-0.8,-0.6);

\node[red!80!black] at (-2,-1.15) {$A_2^-\, e^{-2ik_0 x}$};
\node[red!80!black] at (-2,-1.45) {$2\omega$};

\draw[->,red!80!black] (-1.6,-0.9) -- (-2.2,-0.9);
\node[red!80!black] at (-0.4,-0.85) {$2k_0$};


\draw[domain=0.2:3.5,smooth,variable=\x,red!80!black,thick]
plot ({\x},{0.30*sin(8*\x r)-0.6});

\draw[->,red!80!black,thick] (0.2,-0.6) -- (0.8,-0.6);

\node[red!80!black] at (2,-1.15) {$A_2^+\, e^{2ik_0 x}$};
\node[red!80!black] at (2,-1.45) {$2\omega$};

\draw[->,red!80!black] (1.6,-0.9) -- (2.2,-0.9);
\node[red!80!black] at (0.4,-0.85) {$2k_0$};

\end{tikzpicture}
  \end{center}
  \caption{An idealised medium whose $\chi^{(2)}$ reduces to a
    singular distribution localised at $x=0$.  The fundamental
    field $u_1$ at frequency $\omega$ propagates from the left,
    and the second-harmonic field $u_2$ at $2\omega$ radiates
    outward from the interface on both sides.}
  \label{fig:chi2_simple}
\end{figure}

\begin{figure}[h]
  \begin{center}
    \begin{tikzpicture}[scale=1.2]

\draw[->] (-4,0) -- (4,0) node[right] {$x$};

\fill[orange!80!black] (0,0) circle (0.12);
\node[orange!80!black] at (0.6,0.4) {$\chi^{(2),S}$};

\node at (0,-0.5) {$\delta(x)$};

\draw[blue,thick,decorate,decoration={snake,amplitude=1.2pt,segment length=6pt}]
    (-3,0.8) -- (-0.12,0.08);
\draw[blue,thick,decorate,decoration={snake,amplitude=1.2pt,segment length=6pt}]
    (-3,0.4) -- (-0.12,0.08);
\node[blue] at (-2,1.1) {$\omega$};
\node[blue] at (-2,0.7) {$\omega$};

\draw[red,thick,decorate,decoration={snake,amplitude=1.2pt,segment length=6pt}]
    (0.12,0.08) -- (3,0.8);
\draw[red,thick,decorate,decoration={snake,amplitude=1.2pt,segment length=6pt}]
    (0.12,0.08) -- (-3,-0.8);

\node[red] at (2,1.1) {$2\omega$};
\node[red] at (2,0.8) {$2k_0$};

\node[red] at (-2,-1.1) {$2\omega$};
\node[red] at (-2,-0.8) {$2k_0$};

\node[blue] at (-3.5,0.8) {incoming photons};
\node[red] at (3.5,0.8) {generated photons};

\draw[->,blue] (-2.5,0.6) -- (-2.1,0.2);
\draw[->,red] (1.5,0.6) -- (1.9,0.2);
\draw[->,red] (-1.5,-0.6) -- (-1.9,-0.2);

\end{tikzpicture}
  \end{center}
  \caption{Photon-level picture: two photons at frequency $\omega$
    combine at the surface $\chi^{(2)}$ vertex and produce one
    outgoing photon at frequency $2\omega$ on each side of the
    interface.}
  \label{fig:feynman}
\end{figure}
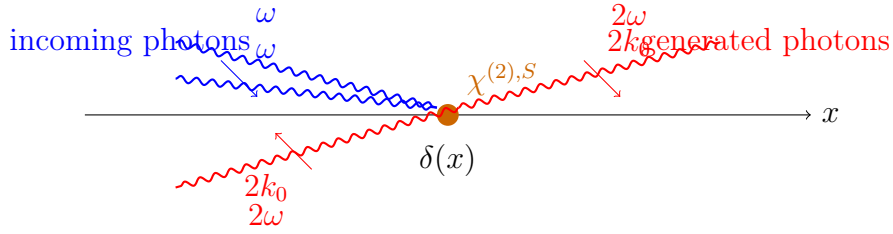

\section{A non-trivial worked example: bulk–surface competition
         in a Yukawa slab}
\label{sec:example_yukawa}

The previous example reduced to a surface at a single interface with
unit refractive index, which is useful pedagogically but hides the
interplay between bulk and surface contributions that is the central
result of the present paper.  We now consider a \emph{slab} of
thickness~$d$, still in a unit-index background, whose second-order
response is governed by the one-dimensional Yukawa kernel.  This
example is explicitly solvable and exhibits a non-trivial competition
between the two surface contributions (one at each interface) and the
distributed bulk source.

\paragraph{Geometry and kernel.}
Let $\Omega = (0,d)$.  We take the scalar 1D Yukawa kernel
\begin{equation}\label{eq:Yukawa_1D}
  \tilde\Delta^{(2)}(Z',Z'')
  = \frac{A_0^{(2)}}{\ell^2}\,e^{-(|Z'|+|Z''|)/\ell},
\end{equation}
where $A_0^{(2)} := \Delta_0^{(2)}\ell^2$ is the bulk amplitude --- an
$\ell$-independent quantity with the dimension of $\chi^{(2)}$.  The
kernel satisfies the intrinsic permutation
symmetry~\eqref{eq:bulk_intrinsic_1D} and is non-centrosymmetric (it
is even under simultaneous reversal, not odd), so that both bulk and
surface SHG are present.

\paragraph{Yukawa kernel coefficients.}
The bulk and surface coefficients are computed directly from the
definitions~\eqref{eq:chi2_jk_bulk} and~\eqref{eq:B0_def}.

\noindent\textit{Bulk.}
\begin{equation}\label{eq:Yukawa_chi00}
  \chi^{(2),\Omega}_{00}
  = \int_{-\infty}^{+\infty}\!\int_{-\infty}^{+\infty}
    \frac{A_0^{(2)}}{\ell^2}\,e^{-(|Z'|+|Z''|)/\ell}\,\diff Z'\,\diff Z''
  = \frac{A_0^{(2)}}{\ell^2}\,(2\ell)^2
  = 4A_0^{(2)}.
\end{equation}

\noindent\textit{Surface.}  Using the symmetry $\tilde\Delta^{(2)}(Z',Z'') =
\tilde\Delta^{(2)}(Z'',Z')$:
\begin{equation}\label{eq:Yukawa_B0}
  B_0
  = \int_0^{+\infty}\!\int_0^{+\infty}
    (Z'+Z'')\,\frac{A_0^{(2)}}{\ell^2}\,e^{-(Z'+Z'')/\ell}\,\diff Z'\,\diff Z''
  = 2A_0^{(2)}\cdot\ell
  = 2A_0^{(2)}\ell.
\end{equation}
The ratio $B_0/\chi^{(2),\Omega}_{00} = \ell/2$ has the dimension of a
length, as it must.

\paragraph{Undepleted pump.}
With unit refractive index everywhere, the pump satisfies
$u_1{}'' + k_0^2 u_1 = 0$ and propagates without
reflection:
\begin{equation}\label{eq:pump_slab}
  u_1(z) = A\,e^{ik_0 z}.
\end{equation}
The pump amplitude at the two interfaces is therefore $u_1(0) = A$
and $u_1(d) = A\,e^{ik_0 d}$.

\paragraph{Second-harmonic equation.}
Let $K := 2k_0$ be the SHG wave number.  Setting $u_2 := u^{2\omega}$,
the SHG equation reads
\begin{equation}\label{eq:SHG_slab}
  u_2'' + K^2 u_2
  = -K^2\!\left[
      \chi^{(2),\Omega}_{00}\,A^2\,e^{iKz}\,\mathbf{1}_{[0,d]}(z)
      + B_0\,A^2\,\delta(z)
      + B_0\,A^2\,e^{iKd}\,\delta(z-d)
    \right].
\end{equation}
The three source terms are: (i) the \emph{bulk source}, distributed
throughout the slab and phase-matched to the pump since $K = 2k_0$
exactly; (ii) the \emph{surface source at $z=0$}, proportional to
$[u_1(0)]^2$; (iii) the \emph{surface source at $z=d$},
proportional to $[u_1(d)]^2 = A^2 e^{iKd}$.

\paragraph{Solution.}
With outgoing-wave conditions, the solution is expressed via the
1D Green's function $G(z-z') = e^{iK|z-z'|}/(2iK)$.

\noindent\textit{Transmitted field} ($z > d$).
Since $K = 2k_0$, the bulk integral simplifies dramatically: for
$z > d > z'$,
\[
  \int_0^d G(z-z')\,e^{iKz'}\,\diff z'
  = \frac{e^{iKz}}{2iK}\int_0^d \diff z' = \frac{d\,e^{iKz}}{2iK}.
\]
For the two surface delta-sources at $z' = 0$ and $z' = d$, the
Green's function gives $e^{iKz}/(2iK)$ and $e^{iK(z-d)} \cdot
e^{iKd}/(2iK) = e^{iKz}/(2iK)$ respectively.  Summing all three
contributions:
\begin{equation}\label{eq:u2_slab_trans}
  u_2(z)
  = \frac{iK}{2}\,A^2\,e^{iKz}\!
    \left[2B_0 + K d\,\chi^{(2),\Omega}_{00}\right], \quad z > d.
\end{equation}
Inserting the Yukawa coefficients~\eqref{eq:Yukawa_chi00}
and~\eqref{eq:Yukawa_B0}:
\begin{equation}\label{eq:u2_slab_Yukawa}
  u_2(z)
  = 4ik_0\,A_0^{(2)}\,\bigl(\ell + 2k_0 d\bigr)\,
    A^2\,e^{2ik_0 z}, \qquad z > d.
\end{equation}

\noindent\textit{Reflected field} ($z < 0$).
The two surface sources contribute coherently at $e^{-iKz}$, while the
bulk source contributes with a $d$-dependent phase factor:
\begin{equation}\label{eq:u2_slab_refl}
  u_2(z)
  = A^2\,e^{-iKz}\!\left[
      \frac{iK B_0}{2}\,\bigl(1 + e^{2iKd}\bigr)
      + \frac{\chi^{(2),\Omega}_{00}}{4}\,\bigl(e^{2iKd}-1\bigr)
    \right], \quad z < 0.
\end{equation}
Unlike the transmitted field, the reflected SHG \emph{oscillates}
with~$d$, because the pump and SHG have the same phase velocity
in vacuum (unit index) so there is no bulk phase mismatch in
transmission, whereas in reflection the two counter-propagating
waves accumulate a relative phase $2Kd$ across the slab.

\paragraph{Surface versus bulk: the crossover length.}
The central result is the factor $(\ell + 2k_0 d)$ in~\eqref{eq:u2_slab_Yukawa}:
\begin{itemize}
  \item The term $\ell$ collects the two surface contributions
    ($2B_0 = 4A_0^{(2)}\ell$)
    and is \emph{independent of~$d$}.
  \item The term $2k_0 d$ collects the bulk contribution
    ($Kd\,\chi^{(2),\Omega}_{00} = 2k_0 d \times 4A_0^{(2)}$)
    and grows \emph{linearly with~$d$} (perfect phase matching).
\end{itemize}
The crossover length at which bulk and surface contribute equally is
\begin{equation}\label{eq:dc_Yukawa}
  d_c := \frac{\ell}{2k_0} = \frac{\ell\lambda}{4\pi}.
\end{equation}
For optical wavelengths ($\lambda \approx 500$\,nm) and a typical
nonlocal range $\ell \approx 0.3$\,nm, one finds $d_c \approx 12$\,pm:
even a monolayer-thick slab ($d \sim \ell$) operates in the bulk-dominated
regime.  The surface contribution nevertheless survives as a
$d$-\emph{independent} offset, accessible experimentally by measuring
the SHG intensity as a function of slab thickness and extrapolating to
$d = 0$.

\paragraph{Limiting cases.}
\begin{itemize}
  \item \textbf{$d \to 0$ (two coincident surfaces).}
    Equation~\eqref{eq:u2_slab_Yukawa} reduces to
    $u_2 = 4ik_0A_0^{(2)}\ell A^2 e^{2ik_0 z}$,
    which is exactly \emph{twice} the result~\eqref{eq:u2_final_example}
    of the previous example (one interface with coefficient
    $\chi_0^{(2),S} = B_0 = 2A_0^{(2)}\ell$, amplitude
    $u_2 = ik_0 B_0 A^2 e^{2ik_0|z|}$, so two such interfaces give
    $2ik_0 B_0 A^2 e^{2ik_0 z}$), confirming internal consistency.
  \item \textbf{$d \gg d_c$.}
    The SHG amplitude grows as $|u_2| \propto 2k_0 d$: phase-matched
    bulk SHG accumulates coherently, independent of the surface geometry.
\end{itemize}

\begin{remark}[Role of unit-index background]
  The unit-index assumption ($n = 1$ inside the slab) is crucial for
  the phase-matching condition $K = 2k_0$, which makes the transmitted
  bulk integral simply proportional to~$d$.  For $n \ne 1$, the pump
  wave number inside the slab becomes $nk_0$, and the bulk source
  oscillates as $e^{2ink_0 z}$ while the SHG propagates at $Kz = 2k_0 z$:
  the coherence length $L_c = \pi/\bigl|2nk_0 - K\bigr| = \lambda/\bigl[4(n-1)\bigr]$
  limits the build-up, and the linear growth in~$d$ becomes oscillatory.
  This Maker-fringe behaviour is standard in nonlinear optics~\cite{Maker1962}; the present
  framework recovers it automatically from the Green's function while
  keeping the surface terms unchanged.
\end{remark}

\section{Surface-only SHG from a centrosymmetric Yukawa slab}
\label{sec:example_centro}

We now treat the complementary case: the same slab geometry
($\Omega = (0,d)$, unit-index background, undepleted pump
$u_1(z) = Ae^{ik_0 z}$), but with a \emph{centrosymmetric}
kernel.  This is the physically standard scenario for SHG at
the surface of a material whose bulk point group contains
inversion (silicon, glass, most cubic crystals): the bulk
second-order response vanishes by symmetry, and SHG is generated
solely by the two interfaces where inversion is broken.

\paragraph{Kernel.}
We take the \emph{odd Yukawa} kernel:
\begin{equation}\label{eq:Yukawa_odd}
  \tilde\Delta^{(2)}(Z',Z'')
  = \frac{A_0^{(2)}}{\ell^3}\,(Z'+Z'')\,
    e^{-(|Z'|+|Z''|)/\ell}.
\end{equation}
This kernel satisfies:
\begin{itemize}
  \item \textbf{Intrinsic permutation symmetry}: $\tilde\Delta^{(2)}(Z',Z'')
    = \tilde\Delta^{(2)}(Z'',Z')$ (since $Z'+Z'' = Z''+Z'$). 
  \item \textbf{Centro-symmetry} (Hypothesis~4 of
    Section~\ref{sec:1D_NL}): $\tilde\Delta^{(2)}(-Z',-Z'')
    = -(Z'+Z'')/\ell \cdot e^{-(|Z'|+|Z''|)/\ell}
    = -\tilde\Delta^{(2)}(Z',Z'')$. 
\end{itemize}

\paragraph{Vanishing of the bulk coefficient.}
Since $Z'\mapsto Z'e^{-|Z'|/\ell}$ is odd, $\int_{\R} Z'e^{-|Z'|/\ell}\diff Z'=0$,
and likewise for $Z''$; hence
\begin{equation}\label{eq:chi00_centro}
  \chi^{(2),\Omega}_{00} = 0.
\end{equation}
There is \textbf{no bulk SHG}.

\paragraph{Surface coefficient.}
Using $\int_0^\infty Z^n e^{-Z/\ell}\diff Z = n!\,\ell^{n+1}$:
\begin{equation}\label{eq:B0_centro}
  B_0^{(\mathrm{c})}
  = \frac{A_0^{(2)}}{\ell^3}
    \int_0^{+\infty}\!\int_0^{+\infty}
    (Z'+Z'')^2\,e^{-(Z'+Z'')/\ell}\,\diff Z'\,\diff Z''
  = \frac{A_0^{(2)}}{\ell^3}\cdot 6\ell^4
  = 6A_0^{(2)}\ell.
\end{equation}

\paragraph{Second-harmonic equation.}
With $K = 2k_0$, the SHG equation now contains \emph{only surface
sources}:
\begin{equation}\label{eq:SHG_centro}
  u_2'' + K^2 u_2
  = -K^2 B_0^{(\mathrm{c})}\,A^2\,
    \bigl[\delta(z) + e^{iKd}\,\delta(z-d)\bigr].
\end{equation}

\paragraph{Solution.}
Via the Green's function $G(z-z') = e^{iK|z-z'|}/(2iK)$:

\noindent\textit{Transmitted field} ($z > d$).
Both surface sources give $e^{iKz}/(2iK)$ for $z > d$:
\begin{equation}\label{eq:u2_centro_trans}
  u_2(z)
  = iK\,B_0^{(\mathrm{c})}\,A^2\,e^{iKz}
  = 12ik_0\,A_0^{(2)}\ell\,A^2\,e^{2ik_0 z},
  \qquad z > d.
\end{equation}
The transmitted SHG amplitude is \textbf{independent of the slab
thickness~$d$}.  This is the hallmark of surface-only SHG in a
centrosymmetric medium: no matter how thick the slab, the two
interfaces contribute equally and coherently in transmission,
while the bulk remains silent.

\noindent\textit{Reflected field} ($z < 0$).
The source at $z=0$ gives $e^{-iKz}/(2iK)$; the source at $z=d$
gives $e^{iKd}\cdot e^{iK(d-z)}/(2iK) = e^{i2Kd}\,e^{-iKz}/(2iK)$:
\begin{equation}\label{eq:u2_centro_refl}
  u_2(z)
  = \frac{iK\,B_0^{(\mathrm{c})}\,A^2\,e^{-iKz}}{2}
    \bigl(1 + e^{2iKd}\bigr)
  = iK\,B_0^{(\mathrm{c})}\,A^2\,
    e^{i(Kd-Kz)}\cos(Kd),
  \qquad z < 0.
\end{equation}
The reflected SHG \textbf{intensity} therefore oscillates as
\begin{equation}\label{eq:Maker_centro}
  |u_2(z)|^2 = K^2\,\bigl[B_0^{(\mathrm{c})}\bigr]^2 A^4\,
               \cos^2(2k_0 d).
\end{equation}
These are the \emph{Maker fringes} of a centrosymmetric slab~\cite{Maker1962}.

\paragraph{Comparison with the non-centrosymmetric case.}
Table~\ref{tab:centro_vs_noncentro} summarises the two regimes.
\begin{table}[ht]
  \centering
  \caption{Transmitted and reflected SHG from a Yukawa slab.}
  \label{tab:centro_vs_noncentro}
  \small
  \begin{tabular}{|l|c|c|}
    \toprule
    & Non-centrosymmetric & Centrosymmetric \\
    \midrule
    Kernel & $\tfrac{A_0^{(2)}}{\ell^2}\,e^{-(|Z'|+|Z''|)/\ell}$ (even)
           & $\tfrac{A_0^{(2)}}{\ell^3}(Z'\!+\!Z'')e^{-(|Z'|+|Z''|)/\ell}$ (odd)\\[4pt]
    Bulk $\chi^{(2),\Omega}_{00}$
           & $4A_0^{(2)} \neq 0$
           & $0$ \\[4pt]
    Surface $B_0$
           & $2A_0^{(2)}\ell$
           & $6A_0^{(2)}\ell$ \\[4pt]
    $|u_2^{\rm trans}|$ & $\propto \ell + 2k_0 d$ (grows with $d$)
                        & $= iKB_0^{(\mathrm{c})}A^2$ (constant in $d$) \\[4pt]
    $|u_2^{\rm refl}|^2$ & complex interference & $\propto \cos^2(2k_0 d)$ \\
    \bottomrule
  \end{tabular}
\end{table}
The $d$-independence of the transmitted SHG and the $\cos^2(2k_0 d)$
pattern of the reflected SHG are the two experimental signatures of
a centrosymmetric bulk with broken surface symmetry.

\section{Conclusion}
\label{sec:conclu_NL}

We have extended the distributional framework of the companion
paper~\cite{Zolla2026nonlocalopticalresponsesurface} to the second-order nonlinear case.  Starting
from the most general tensorial nonlocal second-order constitutive
relation and applying a spatial moment expansion combined with a
distributional thin-layer limit, we have shown that the nonlinear
interfacial response condenses, at leading order, into two scalar
nonlinear surface susceptibilities $\chisNLpar$ and $\chisNLper$,
given by explicit quadrature formulas over the exterior half-space.

\paragraph{Key results.}

Three points deserve emphasis.

\begin{enumerate}

  \item \textbf{Symmetry breaking at the surface generates SHG.}
    For a centrosymmetric bulk medium, the bulk second-order response
    vanishes by symmetry, and the entire second-harmonic signal
    originates from the interfacial layer.  The nonlinear surface
    susceptibilities $\chisNLpar$ and $\chisNLper$ are constructive
    generalizations of the phenomenological surface parameters
    introduced by Bloembergen and coworkers, here derived directly
    from the bulk nonlinear kernel.

  \item \textbf{The marginal integration principle.}
    The nonlinear surface problem reduces recursively to an effective
    linear one: by integrating out one field argument of the nonlinear
    kernel, one obtains a field-dependent effective linear kernel
    $\tilde\Delta^{(2\to 1)}_{ij}$, to which all results
    of~\cite{Zolla2026nonlocalopticalresponsesurface} apply.  This closes the hierarchy at every
    order: the $n$-th order surface susceptibility reduces always to
    a $\chi^{(1)}$ effective problem.

  \item \textbf{Curvature corrections are universal.}
    The same hierarchy as in the linear case --- leading
    $\delta_{\dOmega}$, vanishing first-order correction by
    isotropy, curvature terms at order $\ell^2$ --- holds for the
    nonlinear response.  The nonlinear Mie corrections on a
    nanosphere are of relative order $\ell^2/R_0$, i.e., of the
    same magnitude as their linear counterparts.

\end{enumerate}

\paragraph{Outlook.}

The non-degenerate case $\omega_1 + \omega_2 \to \omega_3$ follows
by desymmetrization: one replaces $\tilde\Delta^{(2)}_{ijk}(\bR',\bR'')$
by its symmetrized version
$\tfrac{1}{2}[\tilde\Delta^{(2)}_{ijk}(\bR',\bR'')
+ \tilde\Delta^{(2)}_{ikj}(\bR'',\bR')]$ and the field $\bE^{+,\omega}$
by two distinct fields $\bE^{+,\omega_1}$ and $\bE^{+,\omega_2}$.
All formulas carry over with this substitution.

The extension to third-order ($\chi^{(3)}$) processes such as
four-wave mixing and two-photon absorption follows the same pattern:
the marginal integration principle reduces the $\chi^{(3)}$ surface
problem to a $\chi^{(2\to 1)}$ effective problem, and thence to a
$\chi^{(1)}$ effective problem by a second marginal integration.
This recursive structure is the deepest result of the present work,
and its full exploitation for $\chi^{(n)}$ will be the subject of a
subsequent paper.

\appendix
\section{Centrosymmetry, polar vectors, and the parity of
         nonlocal kernels}
\label{app:centrosymmetry}

The conditions imposed on the nonlocal kernels
$\tilde\Delta^{(1)}_{ij}$ ($\tilde\Delta^{(1)}_{ij}(-\bR)=\tilde\Delta^{(1)}_{ij}(\bR)$)  and $\tilde\Delta^{(2)}_{ijk}$ ($\tilde\Delta^{(2)}_{ijk}(-\bR',-\bR'')=-\tilde\Delta^{(2)}_{ijk}(\bR',\bR'')$ ) by centrosymmetry are not identical, and understanding why requires
a careful distinction between the geometric nature of the fields
involved.  This appendix makes this distinction precise, and places
it in the framework of differential geometry.

\subsection{Vectors, covectors, and the inversion map}

In the language of differential geometry, the physical fields
appearing in the constitutive relations are not all of the same
geometric type.

\paragraph{The electric field $\bE$ is a covector (1-form).}
It is defined as the gradient of the electric potential $\varphi$:
$\bE = -\nabla\varphi$.  Under the inversion map
$\iota : \br \mapsto -\br$, the potential $\varphi$ is a scalar
(it does not change sign), while its gradient acquires a sign:
\begin{equation}
  \iota^*\bE = -\bE.
\end{equation}
More precisely, $\bE$ is a \emph{polar} covector: it changes sign
under inversion.

\paragraph{The polarization $\bP$ is a polar vector.}
It is defined as an electric dipole moment per unit volume.  A
dipole moment $\mathbf{p} = q\,\mathbf{d}$ involves a charge $q$
(scalar, invariant under inversion) and a displacement vector
$\mathbf{d}$ (polar vector, changes sign under inversion).  Hence:
\begin{equation}
  \iota^*\bP = -\bP.
\end{equation}

\paragraph{The displacement vector $\bR'$.}
The vector $\bR' = \br' - \br$ is a polar vector:
$\iota^*\bR' = -\bR'$.

These three sign changes are the only ingredients needed to
determine the parity of the kernels.

\subsection{Centrosymmetry as an invariance condition}

A material is \emph{centrosymmetric} if its physical properties are
invariant under the inversion map $\iota$.  For the constitutive
relation, this means that if the field configuration
$\bE(\br')$ produces the polarization $\bP(\br)$, then the
inverted field configuration $\tilde\bE(\br') := \bE(-\br')$
must produce the inverted polarization $\tilde\bP(\br) :=
\bP(-\br)$.

Applying this to the general nonlocal relation
\begin{equation}\label{eq:constit_app}
  P_i(\br)
  = \varepsilon_0\int_{\R^3}
    \tilde\Delta^{(1)}_{ij}(\bR')\,E_j(\br+\bR')\,\diff\bR',
\end{equation}
and using $\iota^*\bP = -\bP$, $\iota^*\bE = -\bE$,
$\iota^*\bR' = -\bR'$:
\begin{equation}
  -P_i(\br)
  = \varepsilon_0\int_{\R^3}
    \tilde\Delta^{(1)}_{ij}(-\bR')\,
    \bigl(-E_j(\br+\bR')\bigr)\,\diff\bR'.
\end{equation}
The two minus signs cancel, leaving:
\begin{equation}\label{eq:parity1}
  \tilde\Delta^{(1)}_{ij}(-\bR')
  = \tilde\Delta^{(1)}_{ij}(\bR').
\end{equation}
The linear kernel is \emph{even} under inversion.  This is the
correct mathematical expression of centrosymmetry for a first-order
nonlocal kernel.

\subsection{The second-order case}

For the second-order constitutive relation
\begin{equation}\label{eq:constit2_app}
  P^{(2)}_i(\br)
  = \varepsilon_0\int_{\R^3}\!\int_{\R^3}
    \tilde\Delta^{(2)}_{ijk}(\bR',\bR'')\,
    E_j(\br+\bR')\,E_k(\br+\bR'')\,
    \diff\bR'\,\diff\bR'',
\end{equation}
the same invariance condition gives, using $\iota^*\bP = -\bP$,
$\iota^*\bE = -\bE$:
\begin{equation}
  -P^{(2)}_i(\br)
  = \varepsilon_0\int_{\R^3}\!\int_{\R^3}
    \tilde\Delta^{(2)}_{ijk}(-\bR',-\bR'')\,
    \bigl(-E_j(\br+\bR')\bigr)\,
    \bigl(-E_k(\br+\bR'')\bigr)\,
    \diff\bR'\,\diff\bR''.
\end{equation}
The two factors $(-E_j)(-E_k) = +E_j E_k$ now have the
\emph{same} sign, so the equation becomes:
\begin{equation}
  -P^{(2)}_i(\br)
  = \varepsilon_0\int_{\R^3}\!\int_{\R^3}
    \tilde\Delta^{(2)}_{ijk}(-\bR',-\bR'')\,
    E_j(\br+\bR')\,E_k(\br+\bR'')\,
    \diff\bR'\,\diff\bR''.
\end{equation}
Comparing with~\eqref{eq:constit2_app} yields:
\begin{equation}\label{eq:parity2}
  \tilde\Delta^{(2)}_{ijk}(-\bR',-\bR'')
  = -\tilde\Delta^{(2)}_{ijk}(\bR',\bR'').
\end{equation}
The second-order kernel is \emph{odd} under simultaneous inversion
of both displacement vectors.

\subsection{The general rule}

The pattern is now clear.  For a kernel of order $n$:
\begin{equation}\label{eq:parity_n}
  \tilde\Delta^{(n)}_{ij_1\cdots j_n}
  (-\bR_1,\ldots,-\bR_n)
  = (-1)^{n+1}\,
    \tilde\Delta^{(n)}_{ij_1\cdots j_n}
    (\bR_1,\ldots,\bR_n).
\end{equation}
The sign $(-1)^{n+1}$ arises from the following count: the
left-hand side of the constitutive relation contributes one factor
of $(-1)$ from $\iota^*\bP = -\bP$, while the right-hand side
contributes $n$ factors of $(-1)$ from $\iota^*\bE = -\bE$,
one per field argument.  The net sign imposed on the kernel is
therefore $(-1)^{1+n}$:

\begin{table}[ht]
\centering
\renewcommand{\arraystretch}{1.8}
\begin{tabular}{|c|c|c|c|}
\hline
Order $n$ & Sign $(-1)^{n+1}$ & Kernel parity
  & Physical consequence \\
\hline
$1$ (linear)
  & $+1$ & even
  & Linear response allowed \\
$2$ ($\chi^{(2)}$, SHG)
  & $-1$ & odd
  & Bulk SHG forbidden \\
$3$ ($\chi^{(3)}$, Kerr)
  & $+1$ & even
  & Bulk Kerr effect allowed \\
$4$
  & $-1$ & odd
  & Bulk $\chi^{(4)}$ forbidden \\
\hline
\end{tabular}
\caption{Parity-induced selection rule for the bulk nonlocal kernel
  $\tilde\Delta^{(n)}_{ij_1\cdots j_n}$ in a centrosymmetric medium. The
  required kernel parity under $\bR\mapsto -\bR$ is $(-1)^{n+1}$: even
  orders ($n=2,4,\ldots$) force the bulk kernel to be odd, hence its
  symmetric integral over space vanishes and all even-order bulk
  susceptibilities are forbidden — the classical centrosymmetric selection
  rule.}
\label{tab:parity_rule}
\end{table}

This is the precise mathematical expression of the well-known rule:
in a centrosymmetric medium, all even-order nonlinear susceptibilities
vanish in the bulk.

\subsection{Connection with differential geometry}

The distinction between $\bE$ and $\bP$ has a deeper geometric
origin that is worth making explicit.

\paragraph{$\bE$ as a differential 1-form.}
In the language of differential geometry, the electric field
$\bE = -\diff\varphi$ is naturally a \emph{differential 1-form},
i.e., a section of the cotangent bundle $T^*M$.  Its components
$E_i$ transform as covariant components (covector) under changes of
coordinates, and in particular change sign under the orientation-
reversing map $\iota$.

\paragraph{$\bP$ as a vector density.}
The polarization $\bP$ is defined as a dipole moment per unit volume.
In the language of differential geometry, it is naturally a
\emph{vector density} (a twisted vector field), i.e., a section of
$TM \otimes |\Lambda^3 T^*M|$, where the factor $|\Lambda^3 T^*M|$
accounts for the volume element.  Under the inversion $\iota$, the
volume element $\diff^3\br$ acquires a factor of $(-1)^3 = -1$
(since $\iota^*(\diff x\wedge\diff y\wedge\diff z) =
(-\diff x)\wedge(-\diff y)\wedge(-\diff z) = -\diff x\wedge\diff y
\wedge\diff z$), while the vector part also changes sign.  The two
signs compensate, and one might expect $\bP$ to be invariant.
However, $\bP$ is a \emph{polar} vector density, meaning it is
defined without the absolute value of the volume element: it
transforms as a polar vector and changes sign under $\iota$.

\paragraph{The kernel as a tensor-valued distribution.}
The kernel $\tilde\Delta^{(n)}_{ij_1\cdots j_n}(\bR_1,\ldots,\bR_n)$
is a tensor-valued distribution on $(\R^3)^n$, of type
$(1,n)$ (one contravariant index from $\bP$, $n$ covariant indices
from the $n$ factors of $\bE$).  Under the inversion $\iota$ acting
simultaneously on all arguments, a $(1,n)$ tensor acquires a sign
$(-1)^{1+n}$: one sign from the contravariant index
(from $\bP$, a polar vector), and $n$ signs from the $n$ covariant
indices (from $\bE$, a polar covector).  Centrosymmetry requires
this tensor to be invariant under $\iota$, which forces:
\begin{equation}
  \tilde\Delta^{(n)}(-\bR_1,\ldots,-\bR_n)
  = (-1)^{1+n}\,\tilde\Delta^{(n)}(\bR_1,\ldots,\bR_n),
\end{equation}
recovering~\eqref{eq:parity_n}.

\begin{remark}[Pseudovectors and axial tensors]
  The magnetic field $\bH$ and the magnetic induction $\bB$ are
  \emph{axial} vectors (pseudovectors): they are invariant under
  inversion rather than changing sign.  This is why magnetic
  materials can have odd-order magnetization responses that are
  forbidden by centrosymmetry for polar responses.  The framework
  developed here applies mutatis mutandis to magnetic susceptibility
  kernels, with the sign rule modified accordingly.
\end{remark}

\begin{remark}[Neumann's principle]
  The parity condition~\eqref{eq:parity_n} is a special case of
  Neumann's principle \cite{ZollaWild3} the symmetry group of a physical property
  must include the symmetry group of the medium.  For a
  centrosymmetric medium, the inversion $\iota$ belongs to the
  symmetry group, and the kernel must be invariant under the
  action of $\iota$ on tensor fields of its type.  The parity
  $(-1)^{n+1}$ is precisely this action, and the vanishing of
  bulk $\chi^{(2)}$ is its most celebrated consequence.
\end{remark}

\printbibliography

\end{document}